\documentclass[preprint,10pt]{elsarticle_mod}
\usepackage{url}
\usepackage{amsmath}
\usepackage{amssymb}
\usepackage{amsthm}
\usepackage{graphics}
\usepackage{graphicx}
\usepackage{times}
\usepackage{algorithm,algorithmic}
\usepackage{listings}
\newcommand{\E}[1]{\mathbf{E}\left[#1\right]}

\newtheorem{theorem}{Theorem}

\newtheorem{Lemma}[theorem]{Lemma}

\urldef{\mailsa}\path|rafal.kapelko@pwr.edu.pl|

\begin{document}
\begin{frontmatter}
\title{On the moment distance of Poisson processes 
%On Displacement to the Power of Random Sensors for Coverage in 1D - draft version
%\thanks{Supported by grant nr XXX} %2012/S1????? of the Institute of Mathematics and Computers Science of the Wroc{\l}aw University of Technology} 
}
%\author{}
\author[pwr]{Rafa\l{} Kapelko\corref{cor1}}
\ead{rafal.kapelko@pwr.edu.pl}
%\author[scs]{Evangelos Kranakis\fnref{scsfootnote}}%\corref{cor2}}
%\ead{kranakis@scs.carleton.ca }
%\fntext[scsfootnote]{Research supported in part by NSERC Discovery grant.}
\fntext[pwrfootnote]{Research supported by grant nr 0401/0086/16}
\cortext[cor1]{Corresponding author at: Department of Computer Science,
Faculty of Fundamental Problems of Technology, Wroc{\l}aw University of Science and Technology, 
 Wybrze\.{z}e Wyspia\'{n}skiego 27, 50-370 Wroc\l{}aw, Poland. Tel.: +48 71 320 33 62; fax: +48 71 320 07 51.}
 %\cortext[cor2]{Corresponding author at: School of Computer Science, Carleton University
%1125 Col. By Dr., Ottawa, ON K1S 5B6, Canada. Tel.: 613 520 2600 x 8090; fax:a 613 520 4334.}
\address[pwr]{ Department of Computer Science,
Faculty of Fundamental Problems of Technology, Wroc{\l}aw University of Science and Technology, Poland}
%\address[scs]{School of Computer Science, Carleton University, Ottawa, ON, Canada}
\begin{abstract}
Consider the distance between two i.i.d. and independent Poisson processes with arrival rate $\lambda>0$ and respective arrival times $X_1,X_2,\dots$ and
$Y_1,Y_2,\dots$ on a line. We give a closed analytical formula for the %expected distance to the power $a$
$\E{|X_{k+r}-Y_k|^a}, $ for any integer $k\ge 1, r\ge 0$ and $a\ge 1.$
The expected difference of the arrival times to the power $a$ between two i.i.d. and independent Poisson processes we represent as the  combination of the Pochhammer polynomials. 

Especially, for $r=0$ and any positive integer $a,$ the following identity is valid 
$$
\E{|X_k-Y_k|^a}=\frac{a!}{\lambda^a}\frac{\Gamma\left(\frac{a}{2}+k\right)}{\Gamma(k)\Gamma\left(\frac{a}{2}+1\right)},
$$
where $\Gamma(z)$ is Gamma function. 
\end{abstract}
\end{frontmatter}

\section{Introduction}
The cost of sensor movement has been the subject of interest in computer science research community.
The paper \cite{spa_2013} addresses the expected sum of movement of $n$ identical sensors displaced uniformly and independently at random in the unit interval
to attain the coverage of the unit interval. Further, in  \cite{KK_2016_cube} the authors studied the movement of $n$ sensors with identical $d-$dimensional cube sensing
radius in $d$ dimensions when the cost of movement of sensor is proportional to some (fixed) power $a>0.$ 

Ajtai et al. \cite{ajtai_84} consider the optimal transportation cost for random matchings of bicolored point sets.
The matching theorems for $N$ random variables independently uniformly distributed in the $d-$dimensional unit cube
$[0,1]^d,$ where $d\ge 2$ were investigated in the book \cite{talagrand_2014}.

More importantly, our work is closely related to \cite{dam_2014} where 
the author studied the event distance between two i.i.d. and independent Poisson processes  with arrival rate $\lambda>0$ and respective arrival times
$X_1,X_2,\dots$ and $Y_1,Y_2,\dots$ on a line.
In \cite{dam_2014} the closed formula for the event distances $\E{|X_{k+r}-Y_k|},$ for
any $k\ge 1, r\ge 0$ was derived as the combination of the Pochhammer polynomials.
The following open problem was proposed in \cite{dam_2014} to study more general moments $\E{|X_{i}-Y_j|^a},$  where $a$ is fixed.

We derive a closed formula for the moments $\E{|X_{k+r}-Y_k|^a},$ for
any $k\ge 1, r\ge 0,$ when  $a$ is positive integer number. %and provide asymptotics to real-valued exponents. 
\subsection{Preliminaries}
In this subsection we introduce some basic concepts and recall some useful identities involving 
indefinite and definite integrals, binomial coefficients and special functions which will be useful in the analysis in the next section.\\
We recall the definition of the Pochhammer polynomial \cite{concrete_1994}
\begin{equation}
\label{eq:rising}
x^{(k)} = \begin{cases} 
x(x+1)\dots(x+k-1) & \mbox{for } k\ge 1\\ 
1 &\mbox{for } k=0. %\\
%0 &\mbox{for } k<0
\end{cases}
\end{equation}
The Euler Gamma function $\Gamma(z)=\int_0^{\infty}t^{z-1}e^{-t}dt$ is defined for $z>0.$ Moreover, we have
\begin{equation}
\label{eq:gamma}
\Gamma(n+1)=n!,
\end{equation}
when $n$ is natural number. 

We will use the Legendre duplication formula (see \cite[Identity 5.5.5]{NIST})
\begin{equation}
\label{eq:legendre}
\Gamma(2z)=(2\pi)^{-1/2}2^{2z-\frac{1}{2}}\Gamma(z)\Gamma\left(z+\frac{1}{2}\right).
\end{equation}
Let $X_i$ be the arrival time of the $i-$th event in a Poisson process with arrival rate $\lambda>0$.
We know that the random variable $X_i$ obeys the Gamma distribution with parameters $i\in\mathbb{N_+}, \lambda>0.$
Its probability density function is given by 
$$f_{i,\lambda}(t)=\lambda e^{-\lambda t}\frac{(\lambda t)^{i-1}}{(i-1)!}$$
and $\Pr\left[X_i\ge t\right]=\int_t^{\infty}f_{i,\lambda}(t)dt.$
%(see \cite{kingman,dam_2014,ross_2002}).
Notice that
\begin{equation}
\label{integral_1}
\int_0^{\infty}t^kf_{m,\lambda}(t)dt=\frac{m^{(k)}}{\lambda^k},
\end{equation}
where $m$ is nonnegative integer and $k\in\mathbb{N}$
(see \cite[Chapter 8]{stat_2011}).
Using integration by parts we can derive the following identity
\begin{align}
\nonumber
\int_0^{x}t^kf_{m,\lambda}(t)dt&=\left[-\frac{(m-1+k)!}{(m-1)!\lambda^k}e^{-\lambda t}\sum_{l=0}^{m-1+k}\frac{\left(\lambda t\right)^l}{l!}\right]^x_0\\
\label{integral_a}&=\frac{m^{(k)}}{\lambda^k}\left(1-e^{-\lambda x}\sum_{l=0}^{m-1+k}\frac{(\lambda x)^l}{l!}\right),
\end{align}
where $m$ is nonnegative integer, $k\in\mathbb{N}$ and $\lambda,x>0.$

%\subsection{Estimations when $a$ is an even natural number}
We will use the following binomial identity %For example, see \cite[Identity ??]{concrete_1994}.\\ 
%Let $a,b\in N\cup\{0\}$ and $k\in N.$ Then
\begin{equation}
 \label{eq:binomial1}
 \sum_{j=0}^{a}(-1)^{a-j}\binom{j+k-1}{k-1}\binom{a-j+k-1}{k-1}=
 \begin{cases} \frac{\Gamma\left(\frac{a}{2}+k\right)}{\Gamma(k)\Gamma\left(\frac{a}{2}+1\right)} &\mbox{if } a \equiv 0 \\
0 & \mbox{if } a \equiv 1 \end{cases} \pmod{2}.
% \binom{\frac{a}{2}+i-1}{i-1},
\end{equation}
This identity can be checked using generating functions. Notice that 
\begin{equation}
\label{eq:dusaser}
\frac{1}{(1-z)^k}=\sum_{j\ge 0}\binom{j+k-1}{k-1}z^j
\end{equation}
(see \cite[Table 321, Section 7.2]{concrete_1994}).
Combining Equation (\ref{eq:dusaser}) with the elementary equality
$\frac{1}{(1-z)^k}\frac{1}{(1+z)^k}=\frac{1}{(1-z^2)^k}$ we have
\begin{equation}
\label{eq:mutanter}
\left(\sum_{j\ge 0}\binom{j+k-1}{k-1}z^j\right)\left(\sum_{j\ge 0}\binom{j+k-1}{k-1}(-1)^jz^j\right)=\sum_{j\ge 0}\binom{j+k-1}{k-1}z^{2j}.
\end{equation}
Multiplying together the power series in the left-hand side of Equation (\ref{eq:mutanter})
and equating the coefficient of $z^a$ on both sides we get
\begin{equation}
 \label{eq:binomial1a}
 \sum_{j=0}^{a}(-1)^{a-j}\binom{j+k-1}{k-1}\binom{a-j+k-1}{k-1}=
 \begin{cases} \binom{\frac{a}{2}+k-1}{k-1} &\mbox{if } a \equiv 0 \\
0 & \mbox{if } a \equiv 1 \end{cases} \pmod{2}.
% \binom{\frac{a}{2}+i-1}{i-1},
\end{equation}
Note that, if $\frac{a}{2}\in \mathbb{N_+}$ (see Equation (\ref{eq:gamma})) then 
\begin{equation}
\label{eq:dudaker}
\binom{\frac{a}{2}+k-1}{k-1}=\frac{\left(\frac{a}{2}+k-1\right)!}{(k-1)!\left(\frac{a}{2}\right)!}=\frac{\Gamma\left(\frac{a}{2}+k\right)}{\Gamma(k)\Gamma\left(\frac{a}{2}+1\right)}.
\end{equation}
Putting together (\ref{eq:binomial1a}) and (\ref{eq:dudaker}) we have the desired Formula (\ref{eq:binomial1}).

Finally it is worth pointing out that, applying Formulas (\ref{integral_1}), (\ref{integral_a}), (\ref{eq:binomial1}) in any mathematical software that performs symbolic calculation 
we get the expressions confirming these formulas. Let us consider Identity (\ref{integral_1}).
We use the following command in Mathematica
\begin{verbatim}
Moment[ErlangDistribution[k,l],m]
\end{verbatim}
and get the desired Identity (\ref{integral_1})
\begin{verbatim}
l^(-m)Pochhammer[k,m].
\end{verbatim}
For Identity (\ref{eq:binomial1}) we apply the following command in Mathematica
\begin{verbatim}
Sum[(-1)^(a-j)*Binomial[j+k-1,k-1]
*Binomial[a-j+k-1,k-1],{j,0,a}]
\end{verbatim}
and get
\begin{verbatim}
(1+(-1)^a)Gamma[1/2(a+2k)])/(a*Gamma[a/2]*Gamma[k]).
\end{verbatim}
It is easy to see that for $a \equiv 1 \pmod{2}$ Mathematica confirms Identity (\ref{eq:binomial1}). When $a \equiv 0  \pmod{2}$ 
it is enough to combine together the output of Mathematica with equation $\Gamma\left[\frac{a}{2}+1\right]=\frac{a}{2}\Gamma\left[\frac{a}{2}\right]$
(see Equation (\ref{eq:gamma}) for $n:=\frac{a}{2}$ and $n:=\frac{a}{2}+1$).
%Formulas (\ref{eq:binomial2}), (\ref{integral_a}), (\ref{integral_1}), (\ref{eq:binomial1}).
\subsection{Outline and results of the paper}
We consider the expected difference of the arrival times to the power $a$ between two i.i.d. and independent Poisson processes with arrival rate $\lambda>0$ and respective arrival times $X_1,X_2,\dots$ and
$Y_1,Y_2,\dots$ on a line. 
We give a closed form formula 
$\E{|X_{k+r}-Y_k|^a}, $ for any integer $k\ge 1, r\ge 0,$ when $a$
is positive integer number as the combination of the  Pochhammer polynomials (see Theorem \ref{thm:mainclosedbe}  and Theorem \ref{thm:mainclosedbeolul}).
%{thm:mainclosedbe} 
%More importantly, our work is closely related to the work of \cite{dam_2014} where the author considers
%$\E{|X_{k+r}-Y_k|}, $ for any integer $k\ge 1, r\ge 0.$ Our analysis generalizes the work of \cite{dam_2014} from $a=1$ to
%$a>0.$

Especially, for $r=0,$ the closed analytical formula 
for $\E{|X_{k}-Y_k|^a},$ when  $k\ge 1$ and $a\in \mathbb{N_+}$
was obtained in terms of Gamma functions
(see Theorem \ref{thm:mainclosedbe}  and Theorem \ref{thm:mainclosedbeoa1}).

Here is an outline of the paper.
In Section \ref{tight:sec} we obtain closed formula for event distances to the power $a$ of two i.i.d. and independent Poisson processes, when $a\in \mathbb{N_+}.$
%In Section \ref{sec:appl} we consider the application to sensor networks.
%derive the asymptotic results with application to sensor networks. 
Section \ref{sec:concl} provides conclusions with some open problems.
\section{Main result}
\label{tight:sec}
Consider two i.i.d. and independent Poisson processes with arrival rate $\lambda>0$ and respective arrival times $X_1,X_2,\dots$ and
$Y_1,Y_2,\dots$ on a line. We give a closed analytical formula for the moment distances 
$\E{|X_{k+r}-Y_k|^a}, $ for any integer $k\ge 1, r\ge 0,$ when $a$
is positive integer.% with respect to Pochhammer polynomials
%moments (of each integer power a) of the difference between the kth arrival time in two iid Poisson processes
\subsection{Closed formula when $a$ is a positive even natural number}
\label{tight:even}
We begin with the following lemma which is helpful in the proof of Theorem \ref{thm:mainclosedbe}. 
\begin{Lemma} 
\label{lem:mainclosed} 
Assume that, $a$ is positive even natural number. Let $i\ge 1, k\ge 1.$  
Then
$$\E{|X_i-Y_k|^a}=\frac{1}{{\lambda}^a}\sum_{j=0}^{a}\binom{a}{j}(-1)^{a-j}i^{(j)}k^{(a-j)}.$$
\end{Lemma}
\begin{proof}
As a first step, observe the following formula
\begin{equation}
 \label{eq:firster}
 \E{|X_i-Y_k|^a}=\E{(X_i-Y_k)^a}=\sum_{j=0}^{a}\binom{a}{j}(-1)^{a-j}\E{X_i^j}\E{Y_k^{a-j}}.
 \end{equation}
Applying Identity (\ref{integral_1}) for $m:=i,\, k:=j$ and $m:=k,\,\, k:=a-j$ as well as Definition (\ref{eq:rising}) we deduce that
$$
\E{|X_i-Y_k|^a}
=\sum_{j=0}^{a}\binom{a}{j}(-1)^{a-j}\frac{i^{(j)}}{\lambda^j}\frac{k^{(a-j)}}{\lambda^{a-j}}=\frac{1}{{\lambda}^a}\sum_{j=0}^{a}\binom{a}{j}(-1)^{a-j}i^{(j)}k^{(a-j)}.
$$
This completes the proof of Lemma \ref{lem:mainclosed}. 
\end{proof}
We are now ready to prove the main theorem, when $a$ is a positive even natural number.
\begin{theorem} 
\label{thm:mainclosedbe} Let $a$ be a positive even natural number.
Consider two i.i.d. and independent Poisson processes having identical arrival rate $\lambda >0$ and let $X_1,X_2,\dots$ and $Y_1,Y_2,\dots$
be their arrival times, respectively. The following identities are valid for all $k\ge 1, r\ge 0$
\begin{align*}
\E{|X_{k+r}-Y_k|^a}&=\frac{1}{{\lambda}^a}\sum_{j=0}^{a}\binom{a}{j}(-1)^{a-j}(k+r)^{(j)}k^{(a-j)},\\
\E{|X_k-Y_k|^a}&=\frac{a!}{{\lambda}^a}\frac{\Gamma\left(\frac{a}{2}+k\right)}{\Gamma(k)\Gamma\left(\frac{a}{2}+1\right)}.
\end{align*}
\end{theorem}
\begin{proof}
The first part of the theorem follows immediately from Lemma \ref{lem:mainclosed} with $i=k+r.$
Putting together the first part of the theorem with $r=0,$ Definition (\ref{eq:rising}) and Identity (\ref{eq:binomial1}) we get
\begin{align*}
\E{|X_k-Y_k|^a}&=\frac{1}{{\lambda}^a}\sum_{j=0}^{a}\binom{a}{j}(-1)^{a-j}(k)^{(j)}k^{(a-j)}\\
&=\frac{1}{{\lambda}^a}\sum_{j=0}^{a}\binom{a}{j}(-1)^{a-j}\frac{(k+j-1)!}{(k-1)!}\frac{(k+a-j-1)!}{(k-1)!}\\
&=\frac{a!}{\lambda^a}\sum_{j=0}^{a}(-1)^{a-j}\binom{j+k-1}{k-1}\binom{a-j+k-1}{k-1}\\
&=\frac{\Gamma\left(\frac{a}{2}+k\right)}{\Gamma(k)\Gamma\left(\frac{a}{2}+1\right)}.
\end{align*}
%Note that, if $\frac{a}{2}\in \mathbb{N_+},$ then $\binom{\frac{a}{2}+k-1}{k-1}=\frac{\Gamma\left(\frac{a}{2}+k\right)}{\Gamma(k)\Gamma\left(\frac{a}{2}+1\right)}.$
This is enough to prove Theorem \ref{thm:mainclosedbe}. 
\end{proof}
\subsection{Closed formula when $a$ is an odd natural number}
\label{tight:odd}
It is worthwhile to mention that, when the parameter a in exponent is odd number, 
\textit{it is not so easy} to derive the closed form formula (see Theorem \ref{thm:mainclosedbeolul} and Theorem \ref{thm:mainclosedbeoa1}).

The general strategy of our proof is the following. In computing the moment $\E{|X_i-Y_k|^a}$, we are reduced to computing the moment $\E{|X_i-y_k|^a},$
where $y_k$ is fixed variable, $y_k\in(0,\infty)$
(see (\ref{eq:firstera})).

Then we make an important observation that expectation $\E{|X_i-y_k|^a}$ is equal to the sum of the integrals (\ref{eq:integral_first1}) 
and (\ref{eq:integral_first2}). The first integral is easy to compute, while deriving the second integral is \textit{combinatorially challenging.}

Our analysis of the moment distance proceeds along the following steps.
Firstly, we give Lemma \ref{thm:mainclosedb} and Lemma \ref{thm:mainclosedbeo} which are helpful in the proof of Theorem \ref{thm:mainclosedbeolul}.
Then, Theorem \ref{thm:mainclosedbeoa1} follows from Theorem \ref{thm:mainclosedbeolul} and Lemma  \ref{lem:gamma}. %and the identity $\Gamma(1/2)=\sqrt{\pi}$ we conclude that
\begin{Lemma} 
\label{thm:mainclosedb} 
The following identity is valid for all $i\ge 1, k\ge 1,$ when $a$ is an odd natural number
\begin{align*}
& \E{| X_i- Y_k|^a}=
(-1)\frac{1}{{\lambda}^a}\sum_{j=0}^{a}\binom{a}{j}(-1)^{a-j}i^{(j)}k^{(a-j)}\\
&+
\frac{1}{{\lambda}^a}\sum_{j=0}^{a}\binom{a}{j}(-1)^{a-j}i^{(j)}k^{(a-j)}
\sum_{l=0}^{i+j-1}\binom{k+l-1+a-j}{l}\frac{1}{2^{k+l-1+a-j}}.
\end{align*}
%$$\frac{a!}{{\lambda}^a}\sum_{j=0}^{a}(-1)^{a-j}\binom{i-1+j}{i-1}\binom{k-1+a-j}{k-1}$$
\end{Lemma}
\begin{proof}
As a first step, observe the following formula
\begin{align}
 \nonumber
 \E{|X_i-Y_k|^a}&=\int_{0}^{\infty}\E{|X_i-y_k|^a}f_{k,\lambda}(y_k)dy_k\\
 \label{eq:firstera}&=\int_{0}^{\infty}f_{k,\lambda}(y_k)\E{|X_i-y_k|^a}dy_k.
 \end{align}
Hence, computing the moment $\E{|X_i-Y_k|^a}$ is reduced to computing the moment $\E{|X_i-y_k|^a}.$
Observe that
\begin{align*}
\E{|X_i-y_k|^a}&=\int_{y_k}^{\infty}(t-y_k)^af_{i,\lambda}(t)dt
+\int_{0}^{y_k}(y_k-t)^af_{i,\lambda}(t)dt\\
&=\int_{0}^{\infty}(t-y_k)^af_{i,\lambda}(t)dt
-2\int_{0}^{y_k}(t-y_k)^af_{i,\lambda}(t)dt.
\end{align*}
Therefore,  $\E{|X_i-Y_k|^a}$ is equal to the sum of the following two integrals which we evaluate separately.
\begin{align}
& \label{eq:integral_first1} 
\int_{0}^{\infty}f_{k,\lambda}(y_k)
\int_{0}^{\infty}(t-y_k)^af_{i,\lambda}(t)dtdy_k,\\
& \label{eq:integral_first2} 
(-2)\int_{0}^{\infty}f_{k,\lambda}(y_k)
\int_{0}^{y_k}(t-y_k)^af_{i,\lambda}(t)dt
dy_k.
\end{align}
\textbf{Case of integral (\ref{eq:integral_first1}).}\\ 
Observe that
\begin{align*}
\int_{0}^{\infty}&f_{k,\lambda}(y_k)
\int_{0}^{\infty}(t-y_k)^af_{i,\lambda}(t)dtdy_k\\
&=\int_{0}^{\infty}f_{k,\lambda}(y_k)\E{|X_i-y_k|^a}dy_k=
\E{(X_i-Y_k)^a}.
\end{align*}
After that, the calculation are almost exactly. % the same as in the proof of Lemma \ref{lem:mainclosed}. 

Applying Identity (\ref{integral_1}) and Definition (\ref{eq:rising})  we have
\begin{equation}
\label{eq:last1} 
%\label{eq:integral_first1} 
 \int_{0}^{\infty} f_{k,\lambda}(y_k)
\int_{0}^{\infty}(t-y_k)^a f_{i,\lambda}(t) dtdy_k
=
\frac{1}{{\lambda}^a}\sum_{j=0}^{a}\binom{a}{j}(-1)^{a-j}i^{(j)}k^{(a-j)}.
\end{equation}
\textbf{Case of integral (\ref{eq:integral_first2}).} \\
%Applying identity (\ref{integral_a}) we have
\begin{align*}
%\label{eq:integral_first1} 
&(-2) \int_{0}^{\infty}f_{k,\lambda}(y_k)
\int_{0}^{y_k}(t-y_k)^af_{i,\lambda}(t)dtdy_k\\
&=(-2) \int_{0}^{\infty}f_{k,\lambda}(y_k)
\int_{0}^{y_k}\sum_{j=0}^a\binom{a}{j}(-1)^{a-j}y_k^{a-j}t^jf_{i,\lambda}(t)dtdy_k\\
&=\sum_{j=0}^{a}\binom{a}{j}(-1)^{a-j} A(j),
\end{align*}
where
$$A(j)=(-2)
\int_{0}^{\infty}y_k^{a-j}f_{k,\lambda}(y_k)
\int_{0}^{y_k}t^jf_{i,\lambda}(t)dtdy_k.
$$
Using Identity (\ref{integral_a}) for $x:=y_k,$ $k:=j$ and $m:=i$ we get
$$A(j)=(-2)
\int_{0}^{\infty}y_k^{a-j}f_{k,\lambda}(y_k)
\frac{i^{(j)}}{\lambda^j}\left(1-e^{-\lambda y_k}\sum_{l=0}^{i-1+j}\frac{(\lambda y_k)^l}{l!}\right).
$$
Applying Identity (\ref{integral_1}) for $k:=a-j,\,\,$ $m:=k$ and for $k:=0,\,\,$ $m:=k-1+l+a-j\,\,$ $\lambda:=2\lambda$ we have
$$A(j)=A_1(j)+A_2(j),$$
where
\begin{align*}
A_1(j)&=(-2)\int_{0}^{\infty}y_k^{a-j}f_{k,\lambda}(y_k)\frac{i^{(j)}}{\lambda^j} dy_k=(-2)\frac{i^{(j)}}{\lambda^j}\int_{0}^{\infty}y_k^{a-j}f_{k,\lambda}(y_k)dy_k\\
&=(-2)\frac{i^{(j)}}{\lambda^j}\frac{k^{(a-j)}}{\lambda^{a-j}}=\frac{(-2)}{\lambda^a}i^{(j)}k^{(a-j)},
\end{align*}
\begin{align*}
A_2(j)&=\int_{0}^{\infty}2y_k^{a-j}f_{k,\lambda}(y_k)\frac{i^{(j)}}{\lambda^j}
\sum_{l=0}^{i+j-1}e^{-\lambda y_k}\frac{(\lambda y_k)^l}{l!}dy_k\\
&=\frac{i^{(j)}}{\lambda^a}\sum_{l=0}^{i+j-1}\frac{1}{l!(k-1)!2^{k-1+l+a-j}}\int_0^{\infty}2\lambda e^{-2\lambda y_k}\left(2\lambda y_k\right)^{k-1+l+a-j}dy_k\\
&=\frac{i^{(j)}}{\lambda^a}\sum_{l=0}^{i+j-1}\frac{1}{l!}\frac{(k-1+l+a-j)!}{(k-1)!}\frac{1}{2^{k+l-1+a-j}}\\
&=\frac{i^{(j)}k^{(a-j)}}{\lambda^a}\sum_{l=0}^{i+j-1}\binom{k+a-j+l-1}{l}\frac{1}{2^{k+l-1+a-j}}.
\end{align*}
Therefore, we deduce that
\begin{equation}
 \label{eq:last2}
\sum_{j=0}^{a}\binom{a}{j}(-1)^{a-j} A_1(j)=(-2)\frac{1}{{\lambda}^a}\sum_{j=0}^{a}\binom{a}{j}(-1)^{a-j}i^{(j)}k^{(a-j)},
\end{equation}
\begin{align}
 \sum_{j=0}^{a}\binom{a}{j}(-1)^{a-j} A_2(j) = &
\frac{1}{{\lambda}^a}\sum_{j=0}^{a}\binom{a}{j}(-1)^{a-j}i^{(j)}k^{(a-j)}\nonumber \\
\label{eq:last4}
& \times \sum_{l=0}^{i+j-1}\binom{k+l-1+a-j}{l}\frac{1}{2^{k+l-1+a-j}}.
\end{align} 
Adding Formulas (\ref{eq:last1}), (\ref{eq:last2}) and (\ref{eq:last4}) we derive the desired formula for $\E{|X_i-Y_k|^a},$
when $a$ is  odd natural number.
%\begin{align*}
%(-1)&\frac{1}{{\lambda}^a}\sum_{j=0}^{a}\binom{a}{j}(-1)^{a-j}i^{(j)}k^{(a-j)}\\
%&+
%\frac{1}{{\lambda}^a}\sum_{j=0}^{a}\binom{a}{j}(-1)^{a-j}i^{(j)}k^{(a-j)}
%\sum_{l=0}^{i+j-1}\binom{k+l-1+a-j}{l}\frac{1}{2^{k+l-1+a-j}}.
%\end{align*}
This completes the proof of Lemma \ref{thm:mainclosedb}.
\end{proof}
Now we give a simpler expression for the moment distance of two i.i.d. and independent Poisson processes in the following lemma.
\begin{Lemma} 
\label{thm:mainclosedbeo}
Assume that, $a$ is odd natural number. Let $i\ge 1, k\ge 1.$  
Then
%The following identity is valid for all $i\ge 1, k\ge 1,$ when $a$ is odd natural number:
%Consider two $i.i.d$ Poisson processes having identical arrival rate $\lambda >0$ and let $X_1,X_2,\dots$ and $Y_1,Y_2,\dots$
%be their arrival times, respectively. The following identity is valid for all $i\ge 1,$ $k\ge 1,$ when $a$ is odd natural number:
%$$\E{|X_i-Y_i|^a}=\frac{a!}{\lambda^a}\frac{\Gamma\left(\frac{a}{2}+i\right)}{\Gamma(i)\Gamma\left(\frac{a}{2}+1\right)}$$
\begin{align*}
&\E{|X_i-Y_k|^a}\\
&\,\,\,\,\,=
\left(\sum_{l=k}^{i+a-1}\binom{l+k-1}{l}\frac{1}{2^{l+k-1}}\right)\frac{1}{\lambda^a}\sum_{j=0}^{a}\binom{a}{j}(-1)^{a-j}i^{(j)}k^{(a-j)}\\
&\,\,\,\,\,+
\frac{1}{\lambda^a 2^{i+k-2+a}}
\sum_{l=0}^{a-1}\left(\sum_{j=0}^{l}\binom{a}{j}(-1)^{j}i^{(j)}k^{(a-j)}\right) \binom{i+k+a-1}{i+l}.
\end{align*}
%\frac{a!}{{\lambda}^a}\binom{\frac{a}{2}+i-1}{i-1}.$$
\end{Lemma}
\begin{proof}
%Let 
%\begin{align*}
%B_1(j)&=\binom{a}{j}(-1)^{a-j}i^{(j)}k^{(a-j)},\\
%B_2(j)&=\sum_{l=0}^{i-1+j}\binom{l+k-1+a-j}{l}\frac{1}{2^{l+k-1+a-j}}.
%\end{align*}
Applying Lemma \ref{thm:mainclosedb} we deduce that
\begin{equation}
\label{eq:expection}
\E{|X_i-Y_k|^a}=
(-1)\frac{1}{\lambda^a}\sum_{j=0}^{a}B_1(j)+\frac{1}{\lambda^a}\sum_{j=0}^{a}\binom{a}{j}(-1)^{a-j}i^{(j)}k^{(a-j)}B_2(j).
\end{equation}
where 
$
B_1(j)=\binom{a}{j}(-1)^{a-j}i^{(j)}k^{(a-j)},$
\begin{equation}
\label{eq:startwo}
B_2(j)=\sum_{l=0}^{i-1+j}\binom{l+k-1+a-j}{l}\frac{1}{2^{l+k-1+a-j}}.
\end{equation}
%From identity (\ref{eq:binomial1}) we have
%$\sum_{j=0}^{a}A_1(j)=0.$ Therefore
%\begin{equation}
%\label{eq:expection}
%\E{|X_i-Y_i|^a}=\frac{a!}{\lambda^a}\sum_{j=0}^{a}(-1)^{a-j}\binom{i-1+a-j}{i-1}\binom{i-1+j}{i-1}  A_2(j)
%\end{equation}
Using summation by parts
$$\sum_{l=0}^{i-1+j}g(l+1)(f(l+1)-f(l))=
\sum_{l=0}^{i+j}(g(l)-g(l+1))f(l)+g(i+j+1)f(i+j)-g(0)f(0)$$
for $\,\,f(l)=\frac{-2}{2^{l+k-1+a-j}}\,\,$ and
$\,\,g(l) = \begin{cases} 
\binom{l-1+k-1+a-j}{l-1} & \mbox{for } l\ge 1\\ 
0 &\mbox{for } l=0 
\end{cases}\,\,$ as well as the following basic identity\\
$\binom{l+k-1+a-j}{l}-\binom{l-1+k-1+a-j}{l-1}=\binom{l+k-1+a-j-1}{l}$
we have
\begin{align*}
&\sum_{l=0}^{i-1+j} \binom{l+k-1+a-j}{l}\frac{1}{2^{l+k-1+a-j}}\\
&\,\,\,\,=
\sum_{l=0}^{i+j}\binom{l+k-1+a-(j+1)}{l}\frac{1}{2^{l+k-1+a-(j+1)}}
-\frac{1}{2^{i+k-2+a}}\binom{i+k+a-1}{i+j}.
\end{align*}
Therefore
$$B_2(j)=B_2(j+1)+B_3(j),$$
where
\begin{equation}
\label{eq:db3}
B_3(j)=-\frac{1}{2^{i+k-2+a}}\binom{i+k+a-1}{i+j}.
\end{equation}
Hence, we deduce that
\begin{equation}
 \label{eq:final1}
 B_2(j)=B_2(a)+\sum_{l=j}^{a-1}B_3(l)\,\,\,
\text{for} \,\,\, j\in \{0,1,\dots ,a-1\}.
\end{equation}
Applying Identity (\ref{eq:final1}) to Formula (\ref{eq:expection}) we have
\begin{align}
\nonumber\E{|X_i-Y_k|^a}&=
(B_2(a)-1)\frac{1}{\lambda^a}\sum_{j=0}^{a}\binom{a}{j}(-1)^{a-j}i^{(j)}k^{(a-j)}\\
&+\label{eq:as200}
\frac{1}{\lambda^a}\sum_{j=0}^{a-1}\binom{a}{j}(-1)^{a-j}i^{(j)}k^{(a-j)}\sum_{l=j}^{a-1}B_3(l).
\end{align}
%Using identity (\ref{eq:binomial1}) and 
Using the identity 
$\sum_{j=0}^{m}\binom{m+j}{m}2^{-j}=2^m$
(see \cite[Identity 5.20, p. 167]{concrete_1994}) for $m:=k-1$ as well as Formula (\ref{eq:startwo}) we get
\begin{align}
\nonumber B_2(a)-1&=\sum_{l=0}^{i+a-1}\binom{l+k-1}{l}\frac{1}{2^{l+k-1}}-1\\
\nonumber&=\frac{1}{2^{k-1}}\left(\sum_{l=0}^{k-1}\binom{l+k-1}{k-1}2^{-l}+\sum_{l=k}^{i+a-1}\binom{l+k-1}{l}2^{-l}\right)-1\\
\label{eq:as100}&=\sum_{l=k}^{i+a-1}\binom{l+k-1}{l}\frac{1}{2^{l+k-1}}.
\end{align}
Finally, combining together (\ref{eq:as200}), (\ref{eq:as100}), (\ref{eq:db3}) and changing summation in the second sum in (\ref{eq:as200}) we get the desired result.
%Changing summation in the second sum we get
%\begin{align*}
%\E{|X_i-Y_k|^a}&=
%\left(\sum_{l=k}^{i+a-1}\binom{l+k-1}{l}\frac{1}{2^{l+k-1}}\right)\frac{1}{\lambda^a}\sum_{j=0}^{a}\binom{a}{j}(-1)^{a-j}i^{(j)}k^{(a-j)}\\
%&+
%\frac{1}{\lambda^a 2^{i+k-2+a}}
%\sum_{l=0}^{a-1}\left(\sum_{j=0}^{l}\binom{a}{j}(-1)^{j}i^{(j)}k^{(a-j)}\right) \binom{i+k+a-1}{i+l}.
%\end{align*}
%Putting together identity (\ref{eq:binomial1}) and Theorem \ref{thm:mainclosed} we get
%\begin{equation}
%\label{eq:asym1a}
%\E{|X_i-Y_i|^a}=\frac{a!}{{\lambda}^a}\binom{\frac{a}{2}+i-1}{i-1}.
%\end{equation}
%This completes the proof of Lemma \ref{thm:mainclosedbeo}. 
\end{proof}
We are now ready to give the first main result, when $a$ is an odd natural number.
\begin{theorem} 
\label{thm:mainclosedbeolul} 
Let $a$ be an odd natural number.
Consider two i.i.d. and independent Poisson processes having identical arrival rate $\lambda >0$ and let $X_1,X_2,\dots$ and $Y_1,Y_2,\dots$
be their arrival times, respectively. The following identity is valid for all $r\ge 0,$ $k\ge 1$
%$$\E{|X_i-Y_i|^a}=\frac{a!}{\lambda^a}\frac{\Gamma\left(\frac{a}{2}+i\right)}{\Gamma(i)\Gamma\left(\frac{a}{2}+1\right)}$$
\begin{align*}
&\E{|X_{k+r}-Y_k|^a}\\
&\,\,\,\,\,=
\frac{1}{\lambda^a}\frac{\Gamma\left(k+\frac{1}{2}\right)}{\Gamma\left(\frac{1}{2}\right)\Gamma(k+1)}
\sum_{l=0}^{r+a-1}\frac{(2k)^{(l)}}{(k+1)^{(l)}2^l}
\sum_{j=0}^{a}\binom{a}{j}(-1)^{a-j}(k+r)^{(j)}k^{(a-j)}\\
&\,\,\,\,\,+
\frac{1}{\lambda^a 2^{r-1}}\frac{\Gamma\left(\frac{a}{2}+k\right)}{\Gamma(1/2)\Gamma(k)}
\sum_{l=0}^{a-1}\left(\sum_{j=0}^{l}\binom{a}{j}(-1)^{j}(k+r)^{(j)}k^{(a-j)}\right)\frac{k^{\left(\frac{a+1}{2}\right)(2k+a)^{(r)}}}{k^{(r+l+1)}k^{(a-l)}}.
\end{align*}
%\frac{a!}{{\lambda}^a}\binom{\frac{a}{2}+i-1}{i-1}.$$
\end{theorem}
\begin{proof}
%Let 
%$
%C(k,r,l)=\sum_{j=0}^{l}\binom{a}{j}(-1)^{j}(k+r)^{(j)}k^{(a-j)}.
%$$
Applying Lemma \ref{thm:mainclosedbeo} for $i=k+r$ we deduce that
\begin{align}
\nonumber\E{|X_{k+r}-Y_k|^a}&=
\frac{1}{\lambda^a}\sum_{l=k}^{k+r+a-1}\binom{l+k-1}{l}\frac{1}{2^{l+k-1}}C(k,r,a)\\
\label{eq:dupasek01}&+
\frac{1}{\lambda^a 2^{2k+r-2+a}}
\sum_{l=0}^{a-1}C(k,r,l) \binom{2k+r+a-1}{k+r+l},
\end{align}
where 
\begin{equation}
\label{eq:dupasek02}
C(k,r,l)=\sum_{j=0}^{l}\binom{a}{j}(-1)^{j}(k+r)^{(j)}k^{(a-j)}.
\end{equation}
Using the Legendre duplication formula (\ref{eq:legendre}) for $z=\frac{a-1}{2}+k$ we get
\begin{equation}
 \label{eq:twopoints}
 \Gamma(2k+a-1)=\pi^{-1/2}2^{2k+a-2}\Gamma\left(\frac{a-1}{2}+k\right)\Gamma\left(\frac{a}{2}+k\right).
\end{equation}
Applying Formula (\ref{eq:twopoints}) for $a=1$ and the identity $\Gamma(1/2)=\sqrt{\pi}$ as well as Equation (\ref{eq:gamma}) we derive
$$
2^{-2k+1}\frac{(2k-1)!}{(k-1)!k!}=\frac{\Gamma\left(k+\frac{1}{2}\right)}{\Gamma\left(\frac{1}{2}\right)\Gamma(k+1)}.
$$
Using this and Definition (\ref{eq:rising}) we have
\begin{align}
&\sum_{l=k}^{k+r+a-1}\binom{l+k-1}{l}\frac{1}{2^{l+k-1}}=
\frac{1}{2^{2k-1}}\sum_{l=0}^{r+a-1}\binom{2k-1+l}{k+l}\frac{1}{2^l}\nonumber \\
&= \label{eq:twopointsa} 2^{-2k+1}\sum_{l=0}^{r+a-1}\frac{(2k-1)!(2k)^{(l)}}{(k-1)!k!(k+1)^{(l)}2^l}=
\frac{\Gamma\left(k+\frac{1}{2}\right)}{\Gamma\left(\frac{1}{2}\right)\Gamma(k+1)}\sum_{l=0}^{r+a-1}\frac{(2k)^{(l)}}{(k+1)^{(l)}2^l}.
\end{align}
Combining Definition (\ref{eq:rising}), Equation (\ref{eq:gamma}), the identity $\Gamma(1/2)=\sqrt{\pi}$ and Formula (\ref{eq:twopoints}) we get
\begin{align}
\nonumber\frac{1}{2^{2k+r-2+a}}&\binom{2k+r+a-1}{k+r+l}=\frac{1}{2^{2k+r-2+a}}\frac{(2k+r+a-1)!}{(k+r+l)!(k+a-1-l)!}\\
\nonumber&=\frac{1}{2^{2k+r-2+a}}\frac{\Gamma(2k+a-1)(2k+a-1)(2k+a)^{(r)}}{\Gamma(k)k^{(r+l+1)}(k-1)!k^{(a-l)}}\\
\label{eq:twopointsb}&=\frac{\Gamma\left(\frac{a}{2}+k\right)}{\Gamma(1/2)\Gamma(k)}
\frac{k^{\left(\frac{a+1}{2}\right)}}{2^{r-1}}\frac{(2k+a)^{(r)}}{k^{(r+l+1)}k^{(a-l)}}.
\end{align}
Putting together  (\ref{eq:dupasek01}), (\ref{eq:dupasek02}), (\ref{eq:twopointsa}) and (\ref{eq:twopointsb}) completes the proof of Theorem \ref{thm:mainclosedbeolul}. 
\end{proof}
The next lemma will be helpful in the proof of Theorem \ref{thm:mainclosedbeoa1}. 

The proof of Lemma \ref{lem:gamma}
is technically complicated and the overall strategy is the following. 
Firstly, we define $D(k,a)$ by Formula (\ref{eq:24}). Combing together Identities
(\ref{eq:binomial1}) and $k^{(a-j)}k^{(j)}\binom{a}{j}=k^{(j)}k^{(a-j)}\binom{a}{a-j}$ we represent
$D(k,a)$ by Formula (\ref{eq:twostar}). Then we make an important observation that
$D(k,a)$ represented by Equation (\ref{eq:twostar}) is the polynomial of variable $k$  of degree less than or equal to $\frac{a-1}{2}$.
We prove that all coefficient of $D(k,a)$ expect constant term are equal to zero.
Using the binomial identities (\ref{eq:difficult01}), (\ref{eq:difficult})
and the Legendre duplication formula (\ref{eq:legendre}) we deduce that
%(\ref{eq:legendre}) for $z=\frac{a+1}{2}$
%$\Gamma(2z)=(2\pi)^{-1/2}2^{2z-1/2}\Gamma(z)\Gamma(z+1/2)$
%we deduce that
\begin{equation}
\label{eq:enough}
D(k,a)=\frac{a!\sqrt{\pi}}{2\Gamma\left(\frac{a}{2}+1\right)}\,\,\,\, \text{for each}\,\,\,\, k\in\left\{0,-1,-2\dots,-\frac{a-1}{2}\right\}.
\end{equation}
Therefore, Identity  (\ref{eq:enough}) is enough to prove that $D(k,a)$ is constant equal to $\frac{a!\sqrt{\pi}}{2\Gamma\left(\frac{a}{2}+1\right)}.$

The following Mathematica code can be used to confirm numerically validity of Lemma (\ref{lem:gamma}) %for fixed odd natural number $a$
\begin{lstlisting}[frame=single]
F[a_]:=Sum[Sum[(-1)^j*Pochhammer[k,a-j]*Pochhammer[k,j]
*Binomial[a,j],{j,0,l}]*Pochhammer[k,(a+1)/2]
*(Pochhammer[k,1+l]*Pochhammer[k,a-l])^(-1),{l,0,a-1}]
-(Gamma[a+1]*Pi^(1/2))/(2*Gamma[a/2+1]);
\end{lstlisting}
Then the following command
\begin{lstlisting}[frame=single]
FullSimplify[F[a]]
\end{lstlisting}
gives zero for fixed odd natural parameter $a$.
\begin{Lemma}
\label{lem:gamma}
Assume that, $a$ is an odd natural number. Let $k\ge 1.$ Then
\begin{equation}
 \label{eq:lasttech}
\sum_{l=0}^{a-1}\left(\sum_{j=0}^{l}(-1)^jk^{(a-j)}k^{(j)}\binom{a}{j}\right)
\frac{k^{(\frac{a+1}{2})}}{k^{(1+l)}k^{(a-l)}}
=
\frac{a!\sqrt{\pi}}{2\Gamma\left(\frac{a}{2}+1\right)}.
\end{equation}
\end{Lemma}
\begin{proof} %(Lemma \ref{lem:gamma})
From Identities (\ref{eq:binomial1}) and $k^{(a-j)}k^{(j)}\binom{a}{j}=k^{(j)}k^{(a-j)}\binom{a}{a-j}$
we deduce that
%\begin{equation}
%\label{eq:med}
%\sum_{j=0}^{l}(-1)^jk^{(a-j)}k^{(j)}\binom{a}{j}%&=-\sum_{j=a-l}^{a}k^{(j)}k^{(a-j)}(-1)^{j}\binom{a}{a-j}\nonumber\\
%%&=\label{eq:med}\sum_{j=0}^{a-l-1}k^{(j)}k^{(a-j)}(-1)^{j}\binom{a}{a-j}\nonumber\\
%=\sum_{j=0}^{a-l-1}(-1)^jk^{(a-j)}k^{(j)}\binom{a}{j}.
%\end{equation}
\begin{align}
&\,\,\,\,\,\,\,\,\,\,\,\,\,\,\,\,\,\sum_{j=0}^{l}(-1)^jk^{(a-j)}k^{(j)}\binom{a}{j}=-\sum_{j=a}^{a-l}k^{(j)}k^{(a-j)}(-1)^{j}\binom{a}{a-j}\nonumber\\
=&\label{eq:med}\sum_{j=0}^{a-l-1}k^{(j)}k^{(a-j)}(-1)^{j}\binom{a}{a-j}=\sum_{j=0}^{a-l-1}(-1)^jk^{(a-j)}k^{(j)}\binom{a}{j}.
\end{align}
Let
\begin{equation}
\label{eq:24}
D(k,a):=\sum_{l=0}^{a-1}\left(\sum_{j=0}^{l}(-1)^jk^{(a-j)}k^{(j)}\binom{a}{j}\right)\frac{k^{(\frac{a+1}{2})}}{k^{(1+l)}k^{(a-l)}}.
\end{equation}
Applying Equation (\ref{eq:med}) we deduce that
$$
D(k,a)=D_1(k,a)+D_2(k,a),
$$
where
\begin{align*}
D_1(k,a)&=\sum_{l=0}^{\frac{a-1}{2}}\left(\sum_{j=0}^{l}(-1)^jk^{(a-j)}k^{(j)}\binom{a}{j}\right)
\frac{k^{(\frac{a+1}{2})}}{k^{(1+l)}k^{(a-l)}},\\
D_2(k,a)&=\sum^{a-1}_{l=\frac{a-1}{2}+1}\left(\sum_{j=0}^{a-l-1}(-1)^jk^{(a-j)}k^{(j)}\binom{a}{j}\right)
\frac{k^{(\frac{a+1}{2})}}{k^{(1+l)}k^{(a-l)}}.
\end{align*}
%Applying Equation (\ref{eq:med}) we deduce that
%$$
%D_2(k,a)=\sum_{l=\frac{a-1}{2}+1}^{a-1}\left(\sum_{j=0}^{a-l-1}(-1)^jk^{(a-j)}k^{(j)}\binom{a}{j}\right)
%\frac{k^{(\frac{a+1}{2})}}{k^{(1+l)}k^{(a-l)}}.
%$$
Therefore, we have
\begin{align}
\nonumber D(k,a)&=
\sum_{l=0}^{\frac{a-1}{2}}\sum_{j=0}^{l}\binom{a}{j}(-1)^jk^{(j)}(k+a-l)^{(l-j)}(k+l+1)^{\left(\frac{a-1}{2}-l\right)}\\
\label{eq:twostar} &+
\sum_{l=\frac{a-1}{2}+1}^{a-1}\sum_{j=0}^{a-l-1}\binom{a}{j}(-1)^jk^{(j)}(k+l+1)^{(a-1-j-l)}(k+a-l)^{\left(l-\frac{a-1}{2}\right)}.
\end{align}
Observe that $(k+a-l)^{(l-j)},$  $(k+l+1)^{\left(\frac{a-1}{2}-l\right)}$ are polynomials of variable $k$ for each $j\in\{0,1,\dots ,l\},$
$l\in\{0,1,\dots, \frac{a-1}{2}\}$
and
$(k+l+1)^{(a-1-j-l)},$  $(k+a-l)^{\left(l-\frac{a-1}{2}\right)}$ are polynomials of variable $k$ for each $j\in\{0,1,\dots ,a-1-l\},$
$l\in\{\frac{a-1}{2}+1,\dots,a-1\}.$

Therefore,
$D(k,a)$ can be only the polynomial of variable $k$ of degree \textbf{less than or equal} to $\frac{a-1}{2}.$

We have to prove that for all $k\in \mathbb{N_+}:$ $D(k,a)$ is constant equal to $\frac{a!\sqrt{\pi}}{2\Gamma\left(\frac{a}{2}+1\right)}.$

Hence, to prove Equality (\ref{eq:lasttech}) it remains to obtain the following equality
\begin{align*}
 %\label{eq:lasttechenen}
 D(k,a)&=
\sum_{l=0}^{\frac{a-1}{2}}\sum_{j=0}^{l}\binom{a}{j}(-1)^jk^{(j)}(k+a-l)^{(l-j)}(k+l+1)^{\left(\frac{a-1}{2}-l\right)}\\
&+
\sum_{l=\frac{a-1}{2}+1}^{a-1}\sum_{j=0}^{a-l-1}\binom{a}{j}(-1)^jk^{(j)}(k+l+1)^{(a-1-j-l)}(k+a-l)^{\left(l-\frac{a-1}{2}\right)}\\
&=
\frac{a!\sqrt{\pi}}{2\Gamma\left(\frac{a}{2}+1\right)}\,\,\, \textbf{for each} \,\,\, k=0,-1,-2,\dots,-\frac{a-1}{2}.
\end{align*}
%for each $k=0,-1,-2,\dots,-\frac{a-1}{2}.$\\
Let $b\in\{0,\dots,\frac{a-1}{2}\}.$ Observe that
\begin{align*}
(-b+a-l)^{(l-j)}(-b+l+1)^{\left(\frac{a-1}{2}-l\right)}&=0\,\, \text{for}\,\, 0\le l\le b-1, \,\,0\le j\le l,\\
(-b+l+1)^{(a-1-j-l)}(-b+a-l)^{\left(l-\frac{a-1}{2}\right)}&=0\,\, \text{for}\,\, a-1-(b-1)\le l\le a-1,\\ 
&\,\,\,\,\,\,\,\,\,\,\,\,\,\,\,\,\,\,\,\,\,\, 0\le j\le a-l-1,\\
(-b)^{(j)}&=0\,\,\text{for}\,\, b+1\le j.
\end{align*}
Applying this %and identities $(-b)^j=\frac{b!}{(b-j)!},$ $x^{(j)}=\frac{\Gamma(x+j)}{\Gamma(x)}$  we have
we have
\begin{align*}
D(-b,a)&=\sum_{l=b}^{\frac{a-1}{2}}\sum_{j=0}^{b}\binom{a}{j}(-1)^j(-b)^j(-b+a-l)^{(l-j)}(-b+l+1)^{\left(\frac{a-1}{2}-l\right)}\\
+&
\sum_{l=\frac{a-1}{2}+1}^{a-1-b}\sum_{j=0}^{b}\binom{a}{j}(-1)^j(-b)^j(-b+l+1)^{(a-1-j-l)}(-b+a-l)^{\left(l-\frac{a-1}{2}\right)}\\
=&\sum_{l=b}^{\frac{a-1}{2}}\sum_{j=0}^{b}\binom{a}{j}\frac{b!}{(b-j)!}\frac{(a-j-1-b)!}{(a-l-1-b)!}\frac{\left(\frac{a-1}{2}-b\right)!}{(l-b)!}\\
+&
\sum_{l=\frac{a-1}{2}+1}^{a-1-b}\sum_{j=0}^{b}\binom{a}{j}\frac{b!}{(b-j)!}\frac{(a-j-1-b)!}{(l-b)!}\frac{\left(\frac{a-1}{2}-b\right)!}{(a-l-1-b)!}\\
=&\left(\frac{a-1}{2}-b\right)!b!\left(\sum_{j=0}^{b}\binom{a}{j}\binom{a-1-b-j}{b-j}\right)\sum_{l=b}^{a-1-b}\binom{a-1-2b}{a-l-1-b}.
\end{align*}
Notice that
\begin{equation}
\label{eq:difficult01}
\sum_{l=b}^{a-1-b}\binom{a-1-2b}{a-l-1-b}=\sum_{l=0}^{a-1-2b}\binom{a-1-2b}{l}=2^{a-2b-1}.
\end{equation}
Applying this and the identity 
\begin{equation}
\label{eq:difficult}
 \sum_{j=0}^{b}\binom{a}{j}\binom{a-1-b-j}{b-j}=
 \begin{cases} \frac{2^b}{b!}\prod_{j=1}^{b}(a-(2j-1)) &\mbox{if } b \neq 0 \\
1 & \mbox{if } b=0. \end{cases}
% \binom{\frac{a}{2}+i-1}{i-1},
\end{equation}
(see \cite[Identity 7.17, p. 36]{Gould}) we get
\begin{align*}
D(-b,a)
&=2^{a-1}\left(\frac{a-1}{2}-b\right)!\begin{cases} \frac{2^b}{b!}\prod_{j=1}^{b}(a-(2j-1)) &\mbox{if } b \neq 0 \\
1 & \mbox{if } b=0 \end{cases}\\
%\prod_{j=1}^{b}\left(\frac{a}{2}-j+\frac{1}{2}\right)
&=2^{a-1}\left(\frac{a-1}{2}\right)!.
\end{align*}
%Notice that, using Formula (\ref{eq:difficultt}) in any mathematical software that performs symbolic calculation we get the expressions
%confirming Formula (\ref{eq:difficultt}).
Finally, from the Legendre duplication formula (\ref{eq:legendre}) for $z=\frac{a+1}{2}$
%$\Gamma(2z)=(2\pi)^{-1/2}2^{2z-1/2}\Gamma(z)\Gamma(z+1/2)$
we deduce that
$$
D(-b,a)=\frac{a!\sqrt{\pi}}{2\Gamma\left(\frac{a}{2}+1\right)}\,\,\,\, \text{for all}\,\,\,\, b\in\left\{0,\dots,\frac{a-1}{2}\right\}.
$$
This is enough to prove Lemma \ref{lem:gamma}. 
\end{proof}
Finally, we can prove the second main result of this subsection.
\begin{theorem} 
\label{thm:mainclosedbeoa1} Let $a$ be an odd natural number.
Consider two i.i.d. and independent Poisson processes having identical arrival rate $\lambda >0$ and let $X_1,X_2,\dots$ and $Y_1,Y_2,\dots$
be their arrival times, respectively. The following identity is valid for all $k\ge 1$: 
%$$\E{|X_k-Y_k|^a}=\frac{a!}{\lambda^a}\frac{\Gamma\left(\frac{a}{2}+i\right)}{\Gamma(i)\Gamma\left(\frac{a}{2}+1\right)}$$
$$
\E{|X_k-Y_k|^a}=\frac{a!}{\lambda^a}\frac{\Gamma\left(\frac{a}{2}+k\right)}{\Gamma(k)\Gamma\left(\frac{a}{2}+1\right)}.
$$
\end{theorem}
\begin{proof}
Firstly, combining together Definition (\ref{eq:rising}) and Identity (\ref{eq:binomial1}) we deduce that
\begin{equation}
\label{eq:lasek}
\sum_{j=0}^{a}\binom{a}{j}(-1)^{a-j}k^{(j)}k^{(a-j)}=a!\sum_{j=0}^{a}(-1)^{a-j}\binom{j+k-1}{k-1}\binom{a-j+k-1}{k-1}=0.
\end{equation}
We substitute Equation (\ref{eq:lasek}) into Theorem \ref{thm:mainclosedbeolul} and get
\begin{align*}
& \E{|X_{k}-Y_{k}|^a}\\
&\,\,\,\,\,=\frac{1}{\lambda^a \Gamma\left(1/2\right)2^{-1}}\frac{\Gamma\left(\frac{a}{2}+k\right)}{\Gamma(k)}
\sum_{l=0}^{a-1}\left(\sum_{j=0}^{l}\binom{a}{j}(-1)^{j}(k)^{(j)}k^{(a-j)}\right)\frac{k^{\left(\frac{a+1}{2}\right)}}{k^{(l+1)}k^{(a-l)}}.
\end{align*}
Finally, the result of Theorem \ref{thm:mainclosedbeoa1} follows from  Lemma  \ref{lem:gamma} and the identity $\Gamma(1/2)=\sqrt{\pi}.$ 
%$$
%\E{|X_k-Y_k|^a}=\frac{a!}{\lambda^a}\frac{\Gamma\left(\frac{a}{2}+k\right)}{\Gamma(k)\Gamma\left(\frac{a}{2}+1\right)}.
%$$
%This finishes the proof. 
\end{proof}
%\begin{verbatim}
%F[a_] := 
%Sum[Sum[(-1)^j*Pochhammer[k, a - j]*Pochhammer[k, j]*
%Binomial[a, j], {j, 0, l}]*Pochhammer[k, (a + 1)/2]*
%(Pochhammer[k, 1 + l]*Pochhammer[k, a - l])^(-1), {l, 0, a - 1}];
%\end{verbatim}
%$$
%F[a_{\_}] := Sum[Sum[(-1)^j*Pochhammer[k, a - j]*Pochhammer[k, j]
%*Binomial[a, j], {j, 0, l}]*Pochhammer[k, (a + 1)/2]*(Pochhammer[k, 1 + l]*Pochhammer[k, a - l])^(-1), {l, 0, a - 1}];
%$$
%\section{Conclusion}
%\label{sec:concl}
%In this paper, we studied the expected distance to the power $a$ between Poisson events of two i.i.d. Poisson processes
%with arrival rate $\lambda>0$ and respective arrival times $X_1,X_2,\dots$ and
%$Y_1,Y_2,\dots$ on a line. We obtained a closed form formula for the $\E{|X_{k+r}-Y_k|^a},$ where $k\ge1, r\ge0$ and $a\in N$ 
%and provide asymptotics to real-valuded exponents.
%As a consequence we also derived tight bounds on the cost to the power $b>1$ of a minimum matching between two bicolored random point-sets on a line.
\section{Conclusion}
\label{sec:concl}
In this paper, we studied the expected difference of the arrival times to the power $a$ between Poisson events of two i.i.d. and independent Poisson processes
with arrival rate $\lambda>0$ and respective arrival times $X_1,X_2,\dots$ and
$Y_1,Y_2,\dots$ on a line. 

It is obtained a closed form formula for the $\E{|X_{k+r}-Y_k|^a},$ where $k\ge1, r\ge0$ and $a\in \mathbb{N_+}.$ 
as the combination of the Pochhammer polynomials. 
Especially, for $r=0$ and any positive integer $a,$ 
$\E{|X_k-Y_k|^a}=\frac{a!}{\lambda^a}\frac{\Gamma\left(\frac{a}{2}+k\right)}{\Gamma(k)\Gamma\left(\frac{a}{2}+1\right)},$  
where $\Gamma(z)$ is Gamma function.

It is worthwhile to mention that, when $a$ is odd, it is \textit{combinatorially challenging} to derive the closed form formula.
For $a$ even it is \textit{unsurprisingly easier} to obtain the closed form formula.

It would be interesting for future research to find the closed form formula for the expected difference to the power $a$ of two identical other more general
random variables in $(1)$ dimension as well as $(2)$ in the higher dimension.
%As a consequence we derive the application to sensor networks concerning the expected transportation cost to the power $b>0$ of the bicolored matching.
\bibliographystyle{plain}
\bibliography{refs}
%\newpage
\end{document}